\documentclass[letterpaper, 10 pt, conference]{IEEEconf}  % Comment this line out
\IEEEoverridecommandlockouts                              % This command is only
\usepackage{graphics,graphicx,psfrag,array,eepic}
\usepackage{multirow}% http://ctan.org/pkg/multirow
\usepackage{hhline}% http://ctan.org/pkg/
\usepackage{mathtools}
\usepackage[strict]{changepage}
\usepackage{epsfig}
\usepackage{amsmath}
\usepackage{amssymb}
\usepackage{float}
\restylefloat{table}
\usepackage{multicol}
\usepackage{color}
\usepackage{balance}
\usepackage{cite}
\usepackage{booktabs} % For formal tables
\usepackage{caption,booktabs}
\usepackage[tableposition=top]{caption}
\usepackage{subcaption}
\usepackage{lipsum}
\usepackage{comment}
\usepackage{amssymb}
\usepackage{graphicx}
\usepackage{comment}
\usepackage{bm}
\usepackage{xfrac}
\usepackage{algorithm}
\usepackage[noend]{algpseudocode}
\usepackage{mathtools, cuted}
\usepackage[thinc]{esdiff}
\usepackage{psfrag}
\usepackage{array}
\usepackage{latexsym,theorem}
\usepackage{amsfonts,amssymb,amsbsy}
\usepackage{tikz}
\usetikzlibrary{decorations.shapes,shapes.geometric,shadows,arrows,automata,positioning,calendar,mindmap,backgrounds,scopes,chains,er,3d,matrix,fit}
\usepackage{tikz-cd}
\usepackage{pgfplots}
\usetikzlibrary{intersections, pgfplots.fillbetween}
\DeclareMathOperator*{\argmin}{argmin}

\newcommand{\R}{{\rm I\!R}}
\newcommand{\B}{\mathbb{B}}

\newcommand{\N}{{\rm I\!N}}

\newcommand{\cO}{\mathcal{O}}

\newcommand{\cX}{\mathcal{X}}

\newcommand{\cE}{\mathcal{E}}
\newcommand{\cN}{\mathcal{N}}
\newcommand{\cT}{\mathcal{T}}
\newcommand{\cK}{\mathcal{K}}

\newcommand{\cL}{\mathcal{L}}
\newcommand{\cV}{\mathcal{V}}

\newcommand{\cR}{\mathcal{R}}

\newcommand{\cD}{\mathcal{D}}

\newcommand{\cG}{\mathcal{G}}

\newcommand{\diag}{\operatorname{diag}}

\newtheorem{lemma}{Lemma}
\newtheorem{definition}{Definition}
\newtheorem{theorem}{Theorem}
\newtheorem{proposition}{Proposition}

\newtheorem{remark}{Remark}

\newcommand{\prb}[1]{\Pr\left\{#1\right\}}
\newcommand{\quan}[2]{Q_{#2}{\left(#1\right)}}

\usepackage[british]{babel}
\usepackage{cite}
\usepackage{amsmath,amssymb,amsfonts}
\usepackage{graphicx}
\usepackage{textcomp}
\usepackage{mathrsfs}
\usepackage{url}
\usepackage{bm}
\usepackage{xcolor}
\usepackage{nicefrac}
\allowdisplaybreaks
\usepackage{hyperref}

\usepackage{alphalph,etoolbox}
\usepackage{mathtools,subcaption,float}
\allowdisplaybreaks
\usepackage{pgfplots}
\usepgfplotslibrary{colorbrewer}
\pgfplotsset{compat = 1.15} 
\usetikzlibrary{pgfplots.statistics, pgfplots.colorbrewer} 
\usepackage{pgfplotstable}

\title{%
\normalsize Accepted for presentation at the 64th IEEE Conference on Decision and Control,  Brazil, December 10-12, 2025\\[0.6em]
\LARGE \bf
Conformal Data-driven Control of Stochastic Multi-Agent Systems under Collaborative Signal Temporal Logic Specifications
}

\author{Eleftherios E. Vlahakis$^1$, Lars Lindemann$^2$ and Dimos V. Dimarogonas$^1$
\thanks{This work was supported by the Swedish
Research Council (VR), the Knut \& Alice Wallenberg Foundation (KAW), the Horizon Europe Grant SymAware and the ERC Consolidator Grant LEAFHOUND. $^1$Division of Decision and Control Systems, School of Electrical Engineering and Computer Science, KTH Royal Institute of Technology, 10044, Stockholm, Sweden. Email: {\tt\small\{vlahakis,dimos\}@kth.se}. $^2$Automatic Control Laboratory, ETH Zürich,  Zürich, 8092, Switzerland. Email: {\tt\small llindemann@ethz.ch}. %$^3$School of Electronics, Electrical Engineering and Computer Science, Queen's University Belfast, Northern Ireland, UK. Email: {\tt\small p.sopasakis@qub.ac.uk}
}%
}

\begin{document}

\maketitle
\thispagestyle{empty}
\pagestyle{empty}

%%%%%%%%%%%%%%%%%%%%%%%%%%%%%%%%%%%%%%%%%%%%%%%%%%%%%%%%%%%%%%%%%%%%%%%%%%%%%%%%%%%
\begin{abstract}
We address control synthesis of stochastic discrete-time linear multi-agent systems under jointly chance-constrained collaborative signal temporal logic specifications in a distribution-free manner using available disturbance samples, which are partitioned into training and calibration sets. Leveraging linearity, we decompose each agent’s system into deterministic nominal and stochastic error parts, and design disturbance feedback controllers to bound the stochastic errors by solving a tractable optimization problem over the training data. We then quantify prediction regions (PRs) for the aggregate error trajectories corresponding to agent \textit{cliques}, involved in collaborative tasks, using conformal prediction and calibration data. This enables us to address the specified joint chance constraint via Lipschitz tightening and the computed PRs, and relax the centralized stochastic optimal control problem to a deterministic one, whose solution provides the feedforward inputs. To enhance scalability, we decompose the deterministic problem into agent-level subproblems solved in an MPC fashion, yielding a distributed control policy. Finally, we present an illustrative example and a comparison with \cite{VlahakisCDC24}.
\end{abstract}

% \begin{IEEEkeywords}
% multi-agent systems, additive disturbance, signal temporal logic, sequential optimization, mixed-integer programming  
% \end{IEEEkeywords}

\section{Introduction}
\label{sec:introduction}
Multi-agent systems (MAS) arise in various applications, including robotics and cyber-physical systems, where multiple agents collaborate to accomplish tasks jointly. We use signal temporal logic (STL) to define such specifications \cite{Sun2022}, leveraging Boolean and temporal operators for expressing precisely spatio-temporal constraints \cite{MalerSTL2004, Donze2010}. Under stochastic uncertainty, STL specifications are typically formulated as chance constraints, making STL control synthesis challenging, as chance-constrained problems are generally nonconvex and intractable \cite{FarahaniTAC2019}. Existing methods rely on constraint tightening \cite{FarahaniTAC2019, Safaoui2020, LindemannTAC2022} or analytic techniques \cite{Sadigh2016, Sadigh2018} to provide probabilistic guarantees. These approaches may be limited to Gaussian settings \cite{Sadigh2016}, computationally expensive \cite{Safaoui2020}, or rely on  concentration inequalities and union bound \cite{VlahakisCDC24}, making them unsuitable for general MAS. In this work, we propose a tractable data-driven approach for STL synthesis of stochastic MAS under individual and collaborative STL tasks with probabilistic guarantees.

Data-driven methods offer flexibility in relaxing chance constraints by leveraging available samples and providing guarantees through statistical tools such as conformal prediction (CP) \cite{vovk2005algorithmic}. Using a calibration dataset of \textit{i.i.d.} samples, CP enables the construction of prediction regions for a test sample with a specified probability, without requiring knowledge of the underlying distribution. CP has recently been applied in control \cite{lindemann2024CPsurvey}, including STL-based runtime verification for single-agent schemes \cite{Xin2022, LindemannICCPS23}. %\cite{Xin2022, LindemannICCPS23, Zhao2024robust}
%CP has recently been applied in control settings \cite{lindemann2024CPsurvey}, and in the STL framework, for runtime verification problems, employing surrogate models for STL robustness, dynamic models, and conformal quantile regression \cite{Xin2022, LindemannICCPS23, Zhao2024robust} in single-agent formulations. 
Single-agent STL control synthesis in the presence of uncontrollable agents has been explored in \cite{Yu2024}, providing probabilistic guarantees via CP, while a multi-agent reinforcement learning problem is studied in \cite{Kuipers2024} using CP, albeit without provable assurances.

Here, we study control synthesis for stochastic MAS under collaborative STL specifications, formulated as a chance-constrained optimization. We assume: 1) collaborative tasks are assigned to arbitrary groups of agents, called \textit{cliques}; 2) a joint chance constraint is imposed across all cliques with a desired probability level $1-\theta$; and 3) disturbance trajectory samples are available for each agent (Sec. \ref{sec:Prob_setup}).

Assuming linear agent dynamics, we decompose each system into a deterministic nominal part driven by a feedforward term and a stochastic error part regulated by a disturbance-feedback (DF) controller. To design DF, we introduce \textit{nonconformity scores}—functions of aggregate error trajectories parameterized by DF gains—that capture deviations from nominal behavior within each clique and return the maximum deviation across cliques. The DF gains are then obtained by minimizing a CVaR-based cost of the empirical distribution of these scores, computed from training samples (Sec. \ref{sec:training}).

Next, we quantify prediction regions (PRs) for clique-level aggregate error trajectories by computing the $(1-\theta)$th quantile of the empirical distribution of the proposed scores evaluated using the remaining (calibration) samples. This provides joint probabilistic guarantees across cliques via CP (Sec. \ref{sec:calibration}) and enables us to handle the joint chance constraint through Lipschitz tightening and the computed PRs. Consequently, the original centralized chance-constrained problem can be relaxed to a deterministic one, whose solution yields the agents' feedforward input sequences. To improve scalability, we decompose the deterministic problem into agent-level subproblems, solved in an MPC fashion with a shrinking horizon, leading to a distributed controller with feedback and feedforward elements. The method is illustrated and compared to our previous work \cite{VlahakisCDC24} (Secs.~\ref{sec:stl_control_synthesis} and \ref{sec:example}).

\section{Problem setup}\label{sec:Prob_setup}

\subsection{Notation and Preliminaries}

The sets of real numbers and nonnegative integers are $\R$ and $\N$, 
respectively. Let $N,M\in \N$. Then, $\N_{[0,N]}=\{0,1,\ldots,N\}$.  
% The transpose of $\xi$ is $\xi^\intercal$. 
% The identity matrix is $I_n\in\R^{n\times n}$. 
Let $x_{1},\ldots,x_{n}$ be vectors. 
Then, $x=(x_{1},\ldots,x_{n}) = [x_{1}^\intercal  \;\cdots\;x_{n}^\intercal  ]^\intercal$. 
We denote by $\bm{x}(a:b)=(x(a),\ldots,x(b))$ an aggregate vector consisting of $x(t)$, $t\in \N_{[a,b]}$, representing a trajectory. When $x(t)$, $t\in\N_{[a,b]}$, are random vectors, $\bm{x}(a:b)=(x(a),\ldots,x(b))$ is a random process. Let $x_i(t)$, for $t\in\N_{[0,N]}$ and $i\in\N_{[1,M]}$. Then, $\bm{x}(0:N)=(x(0),\ldots,x(N))$ denotes an aggregate trajectory when 
$x(t)=(x_1(t),\ldots,x_M(t))$, $t\in\N_{[0,N]}$. A random variable $w$ following a distribution $\mathscr{D}_w$ is denoted as 
$w \sim\mathscr{D}_w$ and the expected value of 
$w$ is $\mathbb{E}(w)$. The probability of event $Y$ is $\mathrm{Pr}\{Y\}$. 
% The Minkowski sum and the Pontryagin set difference of 
% $S_1\subseteq \R^n$ and $S_2\subseteq \R^n$ are 
% $S_1\oplus S_2=\{s_1+s_2 \;|\; s_1 \in S_1,\; s_2 \in S_2\}$ 
% and $S_1\ominus S_2 = \{s_1\in S_1\;|\;s_1  +s_2 \in S_1, \forall s_2 \in S_2 \}$, respectively.
$\quan{\mathscr{D}}{\delta}$ is the $\delta$-th quantile of a distribution $\mathscr D$, i.e., for $Z \sim \mathscr D$, $\quan{\mathscr{D}}{\delta}=\inf\{z:\prb{Z\leq z}\geq \delta\}$. %$A$ being a (proper) subset of $B$ is denoted as ($A\subsetneq B$) $A \subseteq B$. The Pontryagin set difference of $S_1, S_2\subseteq \R^n$ is $S_1\ominus S_2 = \{s_1\in \R^n\;|\;s_1  +s_2 \in S_1, \forall s_2 \in S_2 \}$. $\mathrm{int}(X)=\{x\in X\mid \exists e>0,\; \mathbb{B}(e)+x\in X\}$, where $\mathbb{B}(e)=\{x\mid \sqrt{x^\top x}\leq e\}$. The \(i\)th, the smallest, and the largest eigenvalues of a symmetric matrix \(M\) are \(\lambda_i(M)\), \(\lambda_{\min}(M)\), and \(\lambda_{\max}(M)\), respectively. The determinant and trace of $M$ is $\det M$ and $\operatorname{trace} M$. 
The ceiling operator is $\lceil{}\cdot{}\rceil$. The cardinality of a set $\cV$ is $|\cV|$. The symbol $\otimes$ denotes the Kronecker product.

\subsubsection{Conformal Prediction}
% \noindent\textbf{Conformal Prediction:} 
Let $\cR^{(0)},\ldots,\cR^{(k)}$ be \textit{i.i.d.} random variables. We refer to $\cR^{(\varsigma)}$, $\varsigma {\in} \N_{[0,k]}$, as \textit{nonconformity scores}. Given $\theta{\in}(0,1)$, an upper bound $q\in\R$ for $\cR^{(0)}$, which we call a test point, can be obtained as follows. % so that 
% \begin{equation}\label{eq:prob_for_R0}
%     \prb{\cR^{(0)}\leq q} \geq 1-\theta,
% \end{equation}
% where $q$ is computed from the samples $\cR^{(1)},\ldots,\cR^{(k)}$, which we call calibration dataset. Specifically, $q$ may be attained as $q=\quan{\cR^{(1)},\ldots,\cR^{(k)},\infty}{1-\theta}$, which is the $(1-\theta)$th quantile of the empirical distribution $\{\cR^{(1)},\ldots,\cR^{(k)},\infty\}$. Assuming $\cR^{(1)}\leq \cdots \leq \cR^{(k)}$, one can pick $q=\cR^{(p)}$, where $p=\lceil(k+1)(1-\theta)\rceil$, which indicates the $p$th smallest nonconformity score. Note that $q$ is finite with $p\in\N_{[1,k]}$ if $k\geq \lceil(k+1)(1-\theta)\rceil$. This choice of \( q \) ensures that \eqref{eq:prob_for_R0} holds since test point $\cR^{(0)}$ and calibration data $\cR^{(1)}, \ldots, \cR^{(k)}$ are i.i.d. \cite{TibshiraniNeurIPS2019}. This is summarized below. 
\begin{lemma}{\cite[Lemma 1]{TibshiraniNeurIPS2019}}\label{lemma:quantile_lemma}
    If \(\cR^{(0)},\ldots,\cR^{(k)}\) are i.i.d. random variables, then for any $\theta\in(0,1)$, we have 
    \begin{equation}\label{eq:quantile_lemma}
        \prb{{\cR^{(0)}}\leq {\quan{{\cR^{(1)},\ldots,\cR^{(k)},\infty}}{{1}{-}\theta}}}\geq {1}-{\theta.}
    \end{equation}
\end{lemma}

\begin{remark}\label{rem:CP}
    % $q$ may be attained as $q=\quan{\cR^{(1)},\ldots,\cR^{(k)},\infty}{1-\theta}$, which is the $(1-\theta)$th quantile of the empirical distribution $\{\cR^{(1)},\ldots,\cR^{(k)},\infty\}$. 
    (i) Assuming $\cR^{(1)}\leq \cdots \leq \cR^{(k)}$, one can compute $\quan{\cR^{(1)},\ldots,\cR^{(k)},\infty}{1-\theta}$ as the $p$th smallest nonconformity score $\cR^{(p)}$, where $p=\lceil(k+1)(1-\theta)\rceil$. $\cR^{(p)}$ is finite or $p\in\N_{[1,k]}$ if $k\geq \lceil(k+1)(1-\theta)\rceil$. %This choice of \( q \) ensures that \eqref{eq:prob_for_R0} holds since test point $\cR^{(0)}$ and calibration data $\cR^{(1)}, \ldots, \cR^{(k)}$ are i.i.d. \cite{TibshiraniNeurIPS2019}.
    (ii) %The coverage guarantees in \eqref{eq:quantile_lemma} are marginal as the probability is defined over 
    The probability \eqref{eq:quantile_lemma} is defined over the randomness in the test and calibration draws. 
    % The randomness in the draw of test and calibration points $\cR^{(0)}$, $\cR^{(1)}, \ldots, \cR^{(k)}$. 
    Conditional coverage guarantees of the form $\Pr\{\cR^{(0)} \leq C \mid \cR^{(1)}, \ldots, \cR^{(k)}\}$ are in general difficult to obtain.  %However, one can show that the conditional probability is a random variable following a beta distribution centered at $1-\theta$ regardless of $k$ \cite{angelopoulos2022gentle,lindemann2024CPsurvey}. 
    Notably, probably approximately correct coverage guarantees
        $\mathrm{Pr}_c\{\mathrm{Pr}\{\cR^{(0)}\leq \quan{\cR^{(1)},\ldots,\cR^{(k)},\infty}{1-\hat{\theta}}\}\geq 1-\theta\}\geq 1-\beta$ can be obtained by setting $\hat{\theta}=\theta-\sqrt{\frac{\ln{(\sfrac{1}{\beta})}}{2k}}$, with the ``outer" probability $\mathrm{Pr}_c$ taken over the randomness in the calibration draw $\cR^{(1)},\ldots,\cR^{(k)}$, and $\beta\in(0,1)$ \cite{Vovk2012}.
\end{remark}

\subsubsection{Conditional value at risk (CVaR)}
% \noindent\textbf{Conditional value at risk (CVaR):} 
For a random variable $\cR\sim \mathscr D$ and confidence level $(1-\theta)$,  %$\mathrm{VaR}_{1-\theta}$ is defined as
% \begin{equation*}
%        \mathrm{VaR}_{1-\theta}(\cR):=\inf \{\eta \in \R \,|\,\prb{\cR\leq\eta}\geq 1-\theta\},
% \end{equation*}
% that is $\mathrm{VaR}_{1-\theta}(\cR)=\quan{\mathscr D}{1-\theta}$. Then, one can show that
% \begin{equation}\label{eq:VaR}
%      \mathrm{VaR}_{1-\theta}(\cR)\leq q \Leftrightarrow  \prb{\cR\leq q }\geq 1-\theta,
% \end{equation}
% where a bound $q$ can be obtained in a data-driven fashion as in \eqref{eq:quantile_lemma}. Let $\cR(M):\cM \to \R$ be a random variable. Unfortunately, optimizing $\mathrm{VaR}_{1-\theta}(\cR(M))$ is challenging since VaR is typically nonconvex in $M$ even if $\cR(M)$ is a convex function. Alternatively, CVaR of $\cR$ with a confidence level of $(1-\theta)$, denoted as 
$\mathrm{CVaR}_{1-\theta}(\cR)$ is the expected value of $\cR$ in the $\theta$-tail exceeding $\quan{\mathscr D}{1-\theta}$, i.e., $\mathbb{E}\left(\cR\, |\, \quan{\mathscr D}{1-\theta}\leq \cR\right)$.  
% the threshold  $\mathrm{VaR}_{1-\theta}(\cR)$, i.e., 
% \begin{equation*}
%     \mathrm{CVaR}_{1-\theta}(\cR):=\mathbb{E}\left(\cR\, |\, \mathrm{VaR}_{1-\theta}(\cR)\leq \cR\right).
% \end{equation*}
CVaR is a coherent risk measure that meets essential criteria such as convexity and monotonicity. It is typically optimized using standard convex and linear programming techniques, and can be formulated as:
\begin{equation}\label{eq:CVaR:opt}
       \mathrm{CVaR}_{1-\theta}(\cR)=\min_{\eta\in \R}\mathbb{E}\left(\eta+\frac{1}{\theta}(\cR-\eta)_+ \right),
\end{equation}
where $(\cdot)_+=\max\{0,\cdot\}$~\cite{Rockafellar2000CVaR}. %Since CVaR . %, relying on the fact that CVaR provides a tight upper bound for \( \mathrm{VaR}_{1-\theta}(\cR) \)~\cite{Rockafellar2000CVaR}, that is,
% \begin{equation*}
%     \mathrm{VaR}_{1-\theta}(\cR) \leq \mathrm{CVaR}_{1-\theta}(\cR).
% \end{equation*}

\subsubsection{Signal temporal logic}
% \noindent\textbf{Signal temporal logic:}  
We consider STL formulas with standard syntax
\begin{equation}\label{eq:STL_syntax}
    \varphi 
    {}\coloneqq{} 
    \top 
    {}\mid{}
    \pi
    {}\mid{}
    \lnot \phi  
    {}\mid{}
    \phi_1 \wedge \phi_2 
    {}\mid{}
    \phi_1 U_{[t_1,t_2]}\phi_2,
\end{equation}
where $\pi:=(\mu(x)\geq 0)$ is a predicate, $\mu(x):\R^{n_x}\to \R$ is a predicate function of $x\in\R^{n_x}$, and $\phi$, $\phi_1$, and $\phi_2$ are STL formulas, 
which are built recursively using predicates $\pi$, logical operators $\neg$ and $\wedge$, and the 
\textit{until} temporal operator $U$, with $[t_1,t_2]\equiv \N_{[t_1,t_2]}$. 
We omit $\lor$ (\textit{or}), $\lozenge$ (\textit{eventually}) and $\square$ (\textit{always}) 
operators from \eqref{eq:STL_syntax} and the sequel, as these may be defined by \eqref{eq:STL_syntax}, 
e.g., $\phi_1\lor \phi_2 = \neg (\neg \phi_1 \wedge \neg \phi_2)$, 
$\lozenge_{[t_1,t_2]}\phi = \top U_{[t_1,t_2]}\phi$, 
and $\square_{[t_1,t_2]}\phi = \lnot \lozenge_{[t_1,t_2]}\lnot \phi$. 
%We are interested in specifications applied over a finite-time horizon, that is, $t_1$, $t_2$ above are finite.  

We denote by $\bm{x}(t) \models \phi$, $t\in\N$, the satisfaction 
of $\phi$, verified over $\bm{x}(t)=(x(t),x(t+1),\ldots)$. 
The validity of a formula $\phi$ can be verified using Boolean or quantitative semantics: $\bm{x}(t) \models \phi {\iff} \rho^{\phi}(\bm{x}(t))> 0$, where the \textit{robustness function} $\rho^\phi: \R^n\times\cdots\times \R^n \to \R$ quantifies how robustly a signal $\bm{x}(t)$ satisfies a formula $\phi$. Based on the Boolean semantics of STL, the horizon of a formula $\phi$ is roughly the minimum length of a signal that is required to verify $\phi$. We refer to \cite{MalerSTL2004} and \cite{Donze2010} for detailed descriptions of the Boolean and quantitative semantics of STL.

To facilitate the definition of joint STL formulas and to represent groups of agents involved in collaborative tasks, we adopt the concept of \textit{cliques} from graph theory. Let $\cG = (\cV, \cE)$ be an undirected graph, potentially with self-loops and multiple edges, with node set $\cV$ of cardinality $M = |\cV|$ and edge set $\cE$. Let $\nu \subseteq \cV$ with $|\nu| \geq 1$, and define $\cE_{\nu} \subseteq \cE$ as the set of edges connecting the nodes in $\nu$. Then, $(\nu, \cE_{\nu})$ is called a \textit{clique} if it forms a complete subgraph of $\cG$~\cite{Orlin1977}.  %i.e., a complete subgraph of $\cG$, if $\cE_{\nu}$ contains all possible edges that connect nodes $\nu$. 
    % The set of cliques of $\cG$ is defined as $\cK = \{\nu\subseteq \cV\mid (\nu,\cE_{\nu}) \text{ is a complete subgraph of }\cG \}$.

% graph description
% \begin{definition}\label{def:cliques}
%     Let $\cG = (\cV,\cE)$ be an undirected graph potentially containing self-loops and multiple edges with node set $\cV$, cardinality $M=|\cV|$, and edge set $\cE$. Let also $\nu\subseteq \cV$, with $|\nu|\geq1$, and define $\cE_{\nu}\subseteq \cE$ as the set of edges connecting the nodes $\nu$. Then, $\cG'=(\nu,\cE_{\nu})$ is a clique \cite{Orlin1977}, if $\cG'$ is a complete graph.  %i.e., a complete subgraph of $\cG$, if $\cE_{\nu}$ contains all possible edges that connect nodes $\nu$. 
%     The set of cliques of $\cG$ is defined as $\cK = \{\nu\subseteq \cV\mid (\nu,\cE_{\nu}) \text{ is a complete subgraph of }\cG \}$.
%     % \begin{align}
%     %     \cK = \{\nu\subseteq \cV: (\nu,\cE_{\nu}) \text{ is a complete subgraph of }\cG \}.
%     % \end{align}
% \end{definition}
% \begin{remark}
%     According to Definition \ref{def:cliques}, the set of cliques of a connected, undirected graph $\cG = (\cV,\cE)$ with no self-loops contains all edges, called two-vertex cliques, and all complete subgraphs of $\cG$ with cardinality greater than two and less than or equal to $|\cV|$.  
% \end{remark}

Consider a graph $\cG = (\cV,\cE)$ with clique set $\cK$. For simplicity, and with a slight abuse of notation, we denote a clique $(\nu,\cE_\nu)\in\cK$ simply by $\nu$. Let $\nu\in\cK$, with $\nu=(i_1,\ldots,i_{|\nu|})$, and consider vectors $x_{i_j}(t)$, $j\in\N_{[1,|\nu|]}$, with $t\in\N$, Then, $x_\nu(t)=(x_{i_1}(t),\ldots,x_{i_{|\nu|}}(t))$ is an aggregate vector, and $\bm{x}_\nu(t)=(x_\nu(t),x_\nu(t+1),\ldots)$ is an aggregate trajectory. We denote the satisfaction of a collaborative task $\phi_\nu$ defined over $\bm{x}_\nu(t)$ by $\bm{x}_\nu(t)\models \phi_\nu$, or equivalently by $\rho^{\phi_\nu}(\bm{x}_\nu(t))> 0$. 

\subsection{Multi-agent system}\label{sec:MAS}

% \subsection{Dynamics}
\subsubsection{Dynamics} We consider a MAS with $M$ agents, with the $i$th agent following the dynamics
\begin{equation}\label{eq:individual_agent_dynamics}
    x_i(t+1)=A_ix_i(t) + B_iu_i(t) + w_i(t),
\end{equation}
where $x_i(t)\in \R^{n_i}$, $u_i\in\R^{m_i}$, and $w_i(t)\in \R^{n_i}$ are the state, input, and disturbance vectors, respectively, $t\in\N$, the initial condition, $x_i(0)$, is known, $A_i\in\R^{n_i\times n_i}$, $B_i\in\R^{n_i\times m_i}$, with the index $i\in\cV=\{1,\ldots,M\}$. Recall a clique $\nu=(i_1,\ldots,i_{|\nu|})$, where $i_j\in\cV$, with $j\in\N_{[1,|\nu|]}$. By collecting individual state, input, and disturbance vectors, as 
$x_\nu(t)=(x_{i_1}(t),\ldots,x_{i_{|\nu|}}(t))\in \R^{n_\nu}$, 
$u_\nu(t)=(u_{i_1}(t),\ldots,u_{i_{|\nu|}}(t))\in\R^{m_\nu}$, 
and $w_\nu(t)=(w_{i_1}(t),\ldots,w_{i_{|\nu|}}(t))\in \R^{n_\nu}$, respectively, we write the aggregate dynamics of $|\nu|$ agents as
\begin{equation}\label{eq:clique_dynamics}
    x_\nu(t+1) = A_\nu x_\nu(t)+B_\nu u_\nu(t)+w_\nu(t),
\end{equation}
\sloppy
with $A_\nu=\diag(A_{i_1}, \ldots ,A_{i_{|\nu|}})$, $B_\nu=\diag(B_{i_1}, \ldots, B_{i_{|\nu|}})$. When $\nu=(1,\ldots,M)$, the aggregate dynamics of the entire MAS are written as
\begin{equation}\label{eq:MAS}
    x(t+1)=Ax(t)+Bu(t)+w(t),
\end{equation}
where $A=\diag(A_1, \ldots ,A_M)$ and $B=\diag(B_1, \ldots, B_M)$. %Next, we define $\nu$ as a clique indicating a group of agents (or an individual agent) involved in a collaborative (or an individual) STL task.

\subsubsection{STL specification} %Let \(\cV = \{1, \dots, M\}\) denote the set of indices of all agents in MAS. 
The MAS is subject to  %, given by
\begin{equation}
    \phi = \bigwedge_{\nu \in \cK_\phi} \phi_\nu, \label{eq:global_phi}    
\end{equation}
which is a conjunctive STL formula, where each conjunct \(\phi_\nu\) defined over $\bm{x}_\nu(t)$ follows the syntax in \eqref{eq:STL_syntax} and represents a formula that involves the clique $\nu$, indicating all agents in $\nu$ can interact with each other. The set \(\cK_\phi\) collects all these cliques induced by $\phi$, and may include individual agents $(|\nu|=1)$ or group of agents $(1<|\nu|\leq |\cV|)$.   
% The structure of \(\phi\) in \eqref{eq:global_phi} induces an interaction graph \(\cG = (\cV, \cE)\), where \(\cV\) is the set of nodes, and \(\cE = \{ (\nu_i, \nu_j) \mid \nu_i, \nu_j \in \nu, \; \nu \in \cK_\phi \}\) is the set of edges. Note that $\cE$ may include self loops (indicating individual tasks) and multiple edges (indicating that two agents may be jointly involved in more than one collaborative task). %In what follows, we denote by \(\cK_\phi\) the set of cliques induced by \(\phi\), \(\cK_\phi^i = \{\nu \in \cK_\phi \mid  \nu \ni i \}\) the subset of cliques containing agent \(i\), and define $\bar{\cK}_\phi$, such that \(\cK_\phi=\bar{\cK}_\phi \cup \cV\).
Let \(\pi \coloneqq (\mu(y) \geq 0)\) be a predicate in \(\phi\), where \(\mu(y): \R^{n_y} \to \R\). The vector \(y \in \R^{n_y}\) can represent either an individual state vector \(x_i \in \R^{n_i}\), where \(i \in \cV\), or an aggregate state vector \(x_\nu \in \R^{n_\nu}\) that collects the states of agents in clique \(\nu \in \cK_\phi\).

\subsubsection{Disturbance} 
Let \(\bm{w}_i(0{:}N{-}1)=\{w_i(0), w_i(1), \dots, \allowbreak  w_i(N{-}1)\}\) denote the disturbance sequence for the $i$th agent. We assume that the joint distribution $\mathscr{D}_{w_i}$ of the random vectors $w_i(t){\in} \R^{n_i}$, $t{\in}\N_{[0,N-1]}$, is unknown, and instead that a disturbance dataset, \( \cD^{w_i} = \{ \bm{w}_i^{(0)}, \ldots, \bm{w}_i^{(k)} \} \), of \( k+1 \) samples is available, with $\bm{w}_i^{(\varsigma)}=(w_i^{(\varsigma)}(0),\ldots,w_i^{(\varsigma)}(N{-}1))$, $\varsigma\in\N_{[0,k]}$. We assume that the disturbance sequence samples in $\cD^{w_i}$ are i.i.d., where each sample represents one realization of the process.  %We also assume that trajectory samples $\bm{w}_i^{(\varsigma)}$ are independent agent-wise, that is $\bm{w}_i^{(\varsigma)}$, $\bm{w}_j^{(\varsigma)}$ are independent for all $i,j\in\cV$ and $\varsigma\in\N_{[0,k]}$.  
Note that although $\bm{w}_i^{(\varsigma)}$, $\varsigma \in\N_{[0,k]}$, are assumed to be independent, the random vectors $w_i(t)$, $t\in\N_{[0,N-1]}$, may be correlated across time for the $i$th agent. We partition $\cD^{w_i}$ into training and calibration sets: $\cD^{w_i}_{\mathrm{train}} = \{\bm{w}_i^{(k_1+1)}, \ldots, \bm{w}_i^{(k)}\}$ and $\cD^{w_i}_{\mathrm{cal}} = \{\bm{w}_i^{(1)}, \ldots, \bm{w}_i^{(k_1)}\}$. We also use the grouped datasets $\cD^{w_\nu}_{\mathrm{train}}$ and $\cD^{w_\nu}_{\mathrm{cal}}$ over cliques.

\subsection{Problem statement}

We wish to solve the  stochastic optimal control problem
\begin{align}
   &\operatorname*{Min.}_{\substack{\bm{u}(0{:}N{-}1)\\ \bm{x}(0{:}N)}}
        \mathbb{E}
        \left(
            \sum_{i=1}^M\big(\sum_{t=0}^{N-1}(\ell_i(x_i(t), u_i(t))) + V_{f,i}(x_i(N))\big)
        \right) \nonumber \\
     &\mathrm{s.t.~}  x(t+1) = Ax(t)+Bu(t)+w(t),\; t\in\N_{[0,N)}, \nonumber \\
        &\;\;\;\;\;\; \prb{\bm{x}(0:N) \models \phi}\geq 1- \theta, \; \mathrm{with}\; x(0)=x_0, \label{eq:multi_agent_problem} 
        % &\;\;\;\;\;\;  \label{eq:MAS_initial_state_prob}
\end{align}    
where $\ell_i:\R^{n_i}\times \R^{m_i}\to \R$, $V_{f,i}:\R^{n_i}\to \R$, the optimization variables are $\bm{u}(0{:}N{-}1)=(u(0),\ldots,u(N-1))$, $\bm{x}(0{:}N)=(x(0),\ldots,x(N))$, with $u(t)=(u_1(t),\ldots,u_M(t))$, $t\in\N_{[0,N-1]}$, and $x(t)=(x_1(t),\ldots,x_M(t))$, $t\in\N_{[0,N]}$, respectively, $\phi$ is the STL formula in \eqref{eq:global_phi}, to be satisfied by $\bm{x}(0{:}N)$ with a probability $1-\theta$, with $\theta\in(0,1)$, $x_0$ is a known initial condition, and $N$ is the horizon of $\phi$.  %In fact, \eqref{eq:multi_agent_problem} is a planning problem, where minimization is taken over sequences of control actions, rather than over control policies.
Solving the problem directly is challenging due to 1) the (joint) probabilistic constraint in \eqref{eq:multi_agent_problem}, equivalently written as
\begin{equation}\label{eq:clique_chance_STL_constraints}
    \prb{\bm{x}_\nu(0:N)\models \phi_\nu,\; \forall \nu\in\cK_\phi}\geq 1-\theta,
\end{equation}
2) the lack of knowledge about the distribution of the disturbance $w(t)$, and 3) the growing computational complexity with the number of agents. We address challenges 1) and 2) next via a data-driven approach, and 3) thereafter.

% \begin{assumption}\label{ass:global_problem}
% For $x(0)=x_0$ and given $\theta\in(0,1)$, Problem \eqref{eq:multi_agent_problem} is feasible. 
% \end{assumption}

\section{Data-driven uncertainty quantification and feedback design}\label{sec:main_results}

\subsection{Decomposition of dynamics}
Due to the linearity in \eqref{eq:individual_agent_dynamics}, the state of the $i$th agent can be decomposed into a deterministic part, $z_i(t)$, and an error term, $e_i(t)$, i.e., $x_i(t) = z_i(t) + e_i(t)$, with initial conditions $z_i(0) = x_i(0)$ and $e_i(0) = 0$. Consider a clique $\nu \in \cK_\phi$, where $\nu = (i_1, \ldots, i_{|\nu|})$, and define the aggregate vectors $z_\nu(t) = (z_{i_1}(t), \ldots, z_{i_{|\nu|}}(t))$ and $e_\nu(t) = (e_{i_1}(t), \ldots, e_{i_{|\nu|}}(t))$. We consider a control policy comprising disturbance-feedback \cite{Goulart2006} and feedforward elements, given by $u_\nu(t) = \sum_{k=0}^{t-1} \Gamma_\nu^{t,k} w_\nu(k) + v_\nu(t)$, where $\Gamma_\nu^{t,k} = \diag(\Gamma_{i_1}^{t,k}, \ldots, \Gamma_{i_{|\nu|}}^{t,k})$, with $\Gamma_{i_j}^{t,k} \in \R^{m_{i_j} \times n_{i_j}}$ for $j \in \N_{[1,|\nu|]}$. Then, the decomposed dynamics of the agents in $\nu$ are
\begin{subequations}\label{eq:clique_decomposed_dynamics}
    \begin{align}  
    z_\nu(t+1) &= A_\nu z_\nu(t) + B_\nu v_\nu(t),\label{eq:clique_determin_dyn} \\
    e_\nu(t+1) &= A_\nu e_\nu(t) + \sum_{k=0}^{t-1} \Gamma_\nu^{t,k} w_\nu(k) + w_\nu(t).\label{eq:clique_error_dyn}
    \end{align}    
\end{subequations}
As systems in \eqref{eq:clique_determin_dyn} and \eqref{eq:clique_error_dyn} are decoupled, and in view of Proposition~\ref{prop:MAS_chance_constraint} below, we first synthesize the gains $\Gamma_\nu^{t,k}$, $k {\in} \N_{[0,t-1]}$, for $t {\in} \N_{[1,N)}$, to obtain tight prediction regions (see Definition \ref{def:PR}) for the trajectories of the error systems in \eqref{eq:clique_error_dyn}, which will be used to handle uncertainty in \eqref{eq:multi_agent_problem}--\eqref{eq:clique_chance_STL_constraints}, and then, design the feedforward elements in \eqref{eq:clique_determin_dyn} (see Sec. \ref{sec:stl_control_synthesis}). Next, we construct trajectory samples of the error system in \eqref{eq:clique_error_dyn} and use them for feedback synthesis and data-driven provable uncertainty quantification.

% W, inspired by Proposition~\ref{prop:MAS_chance_constraint}, since the systems in \eqref{eq:clique_determin_dyn} and \eqref{eq:clique_error_dyn} are decoupled. We next show how to construct trajectory samples of the error system in \eqref{eq:clique_error_dyn} and use them to design feedback gains and PRs with probabilistic guarantees.

% independently of the one in \eqref{eq:clique_determin_dyn}. In light of the probabilistic constraint in \eqref{eq:multi_agent_problem}-\eqref{eq:clique_chance_STL_constraints}, and given PRs (see definition below) for the error systems in \eqref{eq:clique_error_dyn}, the chance constraint in \eqref{eq:multi_agent_problem}-\eqref{eq:clique_chance_STL_constraints} can be conservatively relaxed to deterministic constraints, as 

\begin{definition}\label{def:PR}
    Let $\bm{y}(a{:}b)\in \R^{n(b-a+1)}$ be a random process, with $y(t){\in} \R^n$, $t{\in}\N_{[a,b]}$. We call the ball $\mathbb{B}(q){\subseteq} \R^{n(b-a+1)}$ a prediction region (PR) of $\bm{y}(a{:}b)$ at probability level $1-\theta$, if $\prb{\bm{y}(a{:}b)\in \B(q)}\geq 1-\theta$.
\end{definition}

% As a closed-loop system driven by the random vector $w(t)$, we will show how to obtain PRs of its trajectory $\bm{e}(0:N)=(e(0),\ldots,e(N))$, with $e(0)=0$, via CP using disturbance sample data. 

% We aim to determine parameters to tighten the robustness functions of the formulas $\phi_\nu$, $\nu\in\cK_\phi$ guided by the size of PRs obtained for the error systems in \eqref{eq:clique_error_dyn}, in light of the probabilistic constraint in \eqref{eq:multi_agent_problem} or equivalently in \eqref{eq:clique_chance_STL_constraints}. The following proposition underpins this approach. 
\begin{proposition}[~\cite{VlahakisCDC24}]\label{prop:MAS_chance_constraint}
Let $\bm{x}(0{:}N){=}\bm{z}(0{:}N)+\bm{e}(0{:}N)$, with $\bm{x}(0{:}N){=}(x(0),\ldots,x(N))$, $\bm{z}(0{:}N){=}(z(0),\ldots,z(N))$ and $\bm{e}(0:N)=(e(0),\ldots,e(N))$. Let $q>0$ be such that $\mathrm{Pr}\{\bm{e}(0:N)\in \mathbb{B}(q) \}\geq 1-\theta$. If $\bm{z}(0:N)+\bm{e}(0:N)\models \phi$ for all $\bm{e}(0:N)\in \mathbb{B}(q)$, then $\mathrm{Pr}\{\bm{x}(0)\models \phi\}\geq 1-\theta$.
\end{proposition}
\begin{proof}
    Define events $Y_x:=\bm{x}(0{:}N)\models \phi$, $Y_e:=\bm{e}(0{:}N)\in \B(q)$, and let $Y_e^\prime$ be the complement of $Y_e$. From the law of total probability, it follows that $\mathrm{Pr}\{Y_x\}=\mathrm{Pr}\{Y_x|Y_e\}\mathrm{Pr}\{Y_e\}+\mathrm{Pr}\{Y_x|Y_e'\}\mathrm{Pr}\{Y_e'\}\geq 1-\theta$, since by assumption, $\mathrm{Pr}\{Y_x|Y_e\}=1$ and $\mathrm{Pr}\{Y_e\}\geq 1-\theta$, and $\mathrm{Pr}\{Y_x|Y_e'\}\mathrm{Pr}\{Y_e'\}\geq 0$.
\end{proof}

\subsection{Error trajectory samples}

% We construct error trajectory samples for each agent, using the available disturbance datasets $\cD^{w_i}$ and the error dynamics in \eqref{eq:clique_error_dyn}, with initial condition $e_\nu(0){=}0$, $\nu{\in}\cK_\phi$. Specifically, 
Let matrices $\bm{A}_i$, $\bm{\Gamma}_i$, and $\bm{B}_i$, $i\in\cV$, be defined as
\begin{align}
\setlength{\arraycolsep}{.8pt}
    \begin{bmatrix}
        I_{n_i} & 0 & \cdots & 0\\
        A_i & I_{n_i} & \ddots & 0 \\
        \vdots & \vdots & \ddots & \vdots\\
        A_i^{N-1} & A_i^{N-2} & \cdots & I_{n_i}
    \end{bmatrix}{,} 
    &
    \setlength{\arraycolsep}{1.pt}
    \begin{bmatrix}
        0 & \cdots & \cdots & \cdots & 0 \vphantom{I_{n_i}}\\
        \Gamma_i^{1,0} & 0 & \cdots & \cdots & 0 \vphantom{I_{n_i}}\\
        \Gamma_i^{2,0} & \Gamma_i^{2,1} & 0 & \cdots & 0 \vphantom{I_{n_i}} \\
        \vdots & \vdots & \ddots & \ddots & \vdots \vphantom{I_{n_i}}\\
        \Gamma_i^{N-1,0} & \Gamma_i^{N-1,1} & \cdots & \Gamma_i^{N-1,N-2} & 0 \vphantom{I_{n_i}}
    \end{bmatrix}, \label{eq:gamma_i}
\end{align}
and $\bm{B}_i=\bm{A}_i(I_N \otimes B_i)$, respectively. Then, for the $\nu$th clique, with $\nu = (i_1, \ldots, i_{|\nu|})$, a dataset of aggregate error trajectory samples $\bm{e}_\nu^{(\varsigma)}(1{:}N) = (e_\nu^{(\varsigma)}(1), \ldots, e_\nu^{(\varsigma)}(N))$ can be constructed as
\begin{subequations}\label{eq:cal_traj_i}
    \begin{align}
        \cD^{e_\nu}&{=}\{\bm{e}_\nu^{(0)}(1{:}N),\ldots,\bm{e}_\nu^{(k)}(1{:}N)\}, \label{eq:cDe_cal} \\
       \bm{e}_\nu^{(\varsigma)}(1{:}N)&{=}(\bm{A}_\nu{+}\bm{B}_\nu\bm{\Gamma}_\nu)\bm{w}_\nu^{(\varsigma)}(0{:}N{-}1), \; \varsigma\in\N_{[0,k]}, \label{eq:ith_error_traj}
    \end{align}
\end{subequations} 
with $\bm{A}_\nu{=}\diag(\bm{A}_{i_1},\ldots,\bm{A}_{i_{|\nu|}})$, $\bm{B}_\nu{=}\diag(\bm{B}_{i_1},\ldots,\bm{B}_{i_{|\nu|}})$, and $\bm{\Gamma}_\nu{=}\diag(\bm{\Gamma}_{i_1},\ldots,\bm{\Gamma}_{i_{|\nu|}})$. In the following, we partition $\cD^{e_\nu}$ into $\cD^{e_\nu}_{\mathrm{train}}$ and $\cD^{e_\nu}_{\mathrm{cal}}$. The samples in $\cD^{e_\nu}_{\mathrm{train}}$ are generated from the disturbance samples in $\cD^{w_\nu}_{\mathrm{train}}$ and are linearly parameterized by the feedback gains $\bm{\Gamma}_\nu$, which facilitates both feedback synthesis and PR scaling over cliques, as detailed in Sec.~\ref{sec:training}. The samples in $\cD^{e_\nu}_{\mathrm{cal}}$ are constructed using $\cD^{w_\nu}_{\mathrm{cal}}$ and the trained feedback gains, and are used in Sec.~\ref{sec:calibration} to obtain PRs with probabilistic guarantees.

\subsection{Data-driven PR scaling and disturbance feedback design}\label{sec:training}

Here, we design the feedback gains $\bm{\Gamma}_\nu$, $\nu{\in}\cK_\phi$, using the training datasets $\cD^{e_\nu}_{\mathrm{train}}$. In view of the joint chance constraint in \eqref{eq:clique_chance_STL_constraints}, we also introduce scaling parameters $C_\nu$, $\nu{\in}\cK_\phi$, to weigh uncertainty across cliques, adjusting the size of PRs for the aggregate trajectories of the systems in \eqref{eq:clique_error_dyn}. To this end, we define nonconformity scores $E^{(\varsigma)}(C,\bm{\Gamma})$, $\varsigma{\in}\N_{[k_1+1,k]}$, parameterized by $\bm{\Gamma}{=}\{\bm{\Gamma}_\nu\}_{\nu\in\cK_\phi}$ and $C{=}\{C_\nu\}_{\nu\in\cK_\phi}$,
\begin{equation}\label{eq:nonconf_score_E_C_gamma}
    E^{(\varsigma)}(C,\bm{\Gamma}))=\max_{\nu\in\cK_\phi} \left(C_\nu\|\bm{e}_\nu^{(\varsigma)}(1{:}N)\|\right),
\end{equation}
with $\bm{e}_\nu^{(\varsigma)}(1{:}N)$ as defined as in \eqref{eq:ith_error_traj}. Motivated by \cite{Cleaveland2024} and our previous work \cite[Sec. IV]{VlahakisLCSS24} for single-agent settings, the feedback gains in $\bm{\Gamma}$ and weights in $C$ can be synthesized by minimizing $\quan{E^{(k_1+1)}(C,\bm{\Gamma}),\ldots,E^{(k)}(C,\bm{\Gamma})}{\hat{\theta}}$, requiring $0 \leq C_\nu \leq 1$, $\nu \in \cK_\phi$, and $\sum_{\nu \in \cK_\phi} C_\nu = 1$, with $\hat{\theta}=(1+\frac{1}{k-k_1-1})(1-\theta)$ (see \cite[Sec. 2.1]{lindemann2024CPsurvey} for details). To handle the complexity of this optimization, arising from both the quantile’s encoding and the multi-agent setup, we propose a tractable CVaR-based over approximation, becoming tighter as the number of samples increases \cite{Rockafellar2000CVaR}, as follows:
\begin{subequations}\label{eq:C_gamma_train_cvar}
\begin{align}
  P:= &\operatorname*{Minimize}_{\substack{\eta\geq 0,Y^{(\varsigma)},C,\bm{\Gamma}}} \eta+\frac{1}{q} \sum_{\varsigma=k_1+1}^k (Y^{(\varsigma)}-\eta)_+\; \mathrm{s.\;t.~} \label{eq:cvar_objective}\\
     &  Y^{(\varsigma)} {\geq} C_\nu\|\bm{e}_\nu^{(\varsigma)}(1{:}N)\|,\;  \forall\nu{\in} \cK_\phi,\;  \varsigma{\in}\N_{(k_1,k]}, \label{eq:max_operator_constraints} \\
     % &\bm{\Gamma}_i\bm{w}_i^{(\varsigma)}(0{:}N{-}1){\in} \cU_i^{N},\;i{\in}\N_{[1,M]},\; \varsigma{\in}\N_{(k_1,k]},\\
     &  0\leq C_\nu\leq 1, \; \nu\in\cK_\phi,\; \sum_{\nu\in\cK_\phi}C_\nu=1,
\end{align}
\end{subequations}
where $q=(k{-}k_1{-}1)(1{-}\hat{\theta})$, and the variables $Y^{(\varsigma)}$ are introduced to address the max operator in \eqref{eq:nonconf_score_E_C_gamma} by the constraints in \eqref{eq:max_operator_constraints}. Next, we denote by $P(C)$ the optimization in \eqref{eq:C_gamma_train_cvar} for fixed weights in $C$ and by $P(\bm{\Gamma})$ the optimization in \eqref{eq:C_gamma_train_cvar} for fixed feedback gains in $\bm{\Gamma}$. Note that both $P(C)$ and $P(\bm{\Gamma})$ are convex. Algorithm \ref{alg:CGamma_train} summarizes an efficient procedure, known as a block-coordinate descent algorithm, for solving \eqref{eq:C_gamma_train_cvar} iteratively, which is a widely used heuristic performing well in practice for a small number of iterations.
\begin{algorithm}
\caption{Iterative procedure for solving \eqref{eq:C_gamma_train_cvar}}
\label{alg:CGamma_train}
\begin{algorithmic}[1]
% \State Given problem \eqref{eq:multi_agent_problem_tight}
\State Set $C^0{=}\{C_\nu{=}0 \; \textnormal{for}\; |\nu|{>}1,\;C_i{=}\frac{1}{M} \;\textnormal{for}\; i{\in}\cV\}$
\For{$\tau$ in $1:\tau_\mathrm{max}$}
\State Solve $P(C^{\tau-1})$ and return $\bm{\Gamma}^\tau$
\State Solve $P(\bm{\Gamma}^\tau)$ and return $C^\tau$. 
\EndFor
\State \textbf{return} $(C^\ast\leftarrow C^{\tau_\mathrm{max}},\;\bm{\Gamma}^\ast\leftarrow\bm{\Gamma}^{\tau_\mathrm{max}})$ \label{algo:return}
\end{algorithmic}
\end{algorithm}

\subsection{Error prediction regions}\label{sec:calibration}

Given the feedback gains in \( \bm{\Gamma}^\ast{ = }\{\bm{\Gamma}_\nu^\ast\}_{\nu\in\cK_\phi} \) and weights in \( C^\ast {=} \{ C_\nu^\ast \}_{\nu\in\cK_\phi} \), obtained via Algorithm \ref{alg:CGamma_train}, we now derive PRs for the aggregate error trajectories in cliques \( \nu \in \cK_\phi \) with the desired confidence level as follows.

\begin{proposition}\label{prop:prediction_regions}
Let $\bm{\Gamma}^\ast$ and $C^\ast$ contain the feedback gains and weights, respectively, obtained via Algorithm \ref{alg:CGamma_train}. Construct dataset $\cD^{e_\nu}_\mathrm{cal}$ as in \eqref{eq:cal_traj_i}, using $\cD^{w_\nu}_\mathrm{cal}$ and $\bm{\Gamma}^\ast$, and compute $E^{(\varsigma)}(C^\ast,\bm{\Gamma}^\ast){=}\max_{\nu\in\cK_\phi} \left(C_\nu^\ast\|\bm{e}_\nu^{(\varsigma)}(1{:}N)\|\right)$, for $\varsigma{\in}\N_{[0,k_1]}$,
% \begin{equation}\label{eq:nonconf_scores_prop}
%     E^{(\varsigma)}(C^\ast,\bm{\Gamma}^\ast)=\max_{\nu\in\cK_\phi} \left(C_\nu^\ast\|\bm{e}_\nu^{(\varsigma)}(1{:}N)\|\right),\, \varsigma{\in}\N_{[0,k_1]},
% \end{equation}
% where $\varsigma\in\N_{[0,k_1]}$, and compute 
and $q{=}\quan{E^{(1)}(C^\ast,\bm{\Gamma}^\ast),\ldots, E^{(k_1)}(C^\ast,\bm{\Gamma}^\ast),\infty}{1-\theta}$. 
% \begin{equation}\label{eq:quantile_q_prop}
% \end{equation}
Then, 
 \begin{align}
   \prb{\bm{e}^{(0)}_\nu(1{:}N) \in \B(\sfrac{q}{C_\nu^\ast}), \; \forall \nu\in\cK_\phi} \geq 1 - \theta. 
    \label{eq:chance_constrained_error_bounds}
\end{align}
\end{proposition}%
\begin{proof}
Since \( \{E^{(0)}(C^\ast,\bm{\Gamma}^\ast),\ldots, E^{(k_1)}(C^\ast,\bm{\Gamma}^\ast)\} \) is a set consisting of i.i.d. random variables, Lemma \ref{lemma:quantile_lemma} implies that  
\begin{equation}\label{eq:E0_prop}
    \prb{E^{(0)}(C^\ast,\bm{\Gamma}^\ast) \leq q} \geq 1-\theta.
\end{equation}  
By the definition of \( E^{(0)}(C^\ast,\bm{\Gamma}^\ast) \), this directly implies that $\prb{\max_{\nu\in\cK_\phi}\left(C_\nu^\ast\|\bm{e}_\nu^{(0)}(1{:}N)\|\right)\leq q} \geq 1-\theta$, or $\prb{\bm{e}_\nu^{(0)}(1{:}N)\in \B(\sfrac{q}{C_\nu^\ast}),\; \forall \nu\in\cK_\phi} \geq 1-\theta$, since the \(\max\) operator ensures that the bound holds uniformly for all \( \nu \in \cK_\phi \), completing the proof.
\end{proof}

%%%%%%%%%%%%%%%%%%%%%%%%%%%%%%%%%%%%%%%%%%%%%%%%%%%%%%%%%%%%%%%%%%%%%%%%%%%%%%%%%%%%%%%%%%%%%

\section{STL Control Synthesis}\label{sec:stl_control_synthesis}

Motivated by Prop.~\ref{prop:MAS_chance_constraint}, we will relax \eqref{eq:multi_agent_problem} into a deterministic problem by tightening STL constraints via our data-driven uncertainty quantification in Proposition \ref{prop:prediction_regions} and solve it through a scalable decomposition.

\subsection{STL Tightening}

We first tighten the robustness function of the formula in \eqref{eq:global_phi} using the Lipschitz constants of its subformulas and the PRs from Prop.~\ref{prop:prediction_regions}. The following two lemmas lead to the relaxed formulation in Theorem~\ref{thm:relaxed_prob}.

% \begin{assumption}\label{ass:predicate} 
%     All predicate functions appearing in the multi-agent STL formula $\phi$ are Lipschitz continuous.
% \end{assumption}

\begin{lemma}{\cite[Prop. 1]{KordabadTAC2025}}\label{lemma:lipschitz}  For any STL formula $\phi$ with Lipschitz continuous predicate functions, the robustness function $\rho^{\phi}(\bm{z}(0{:}N){+}\bm{e}(0{:}N))$ is Lipschitz continuous, with Lipschitz constant $L_{\phi}$ obtained as the maximum Lipschitz constant of the predicate functions appearing in $\phi$.
\end{lemma}
% \begin{proof}
%     See .
% \end{proof}

\begin{lemma}\label{lemma:tightening_single_formula}
Let trajectory $\bm{x}(0{:}N){=}\bm{z}(0{:}N){+}\bm{e}(0{:}N)$, with $\prb{\bm{e}(0{:}N){\in} \B(q/C)}{\geq} 1{-}\theta$, and consider STL formula $\phi$. Let $L_\phi$ be the Lipschitz constant of the robustness function $\rho^{\phi}(\bm{z}(0{:}N){+}\bm{e}(0{:}N))$. If $\rho^\phi(\bm{z}(0{:}N)){\geq} L_\phi   \frac{q}{C}$,  then $\prb{\rho^\phi(\bm{x}(0{:}N)){\geq} 0}{\geq }1{-}\theta$.  
\end{lemma}
\begin{proof}
    % Since $\|\bm{e}(0:N)\|=\left(\sum_{t=0}^N\|e(t)\|^2\right)^{\frac{1}{2}}$ and $\left(\sum_{t\in\N_{[1,N]}}\left(\frac{q}{C}\right)^2\right)^{\frac{1}{2}}= \frac{q}{C}$, 
    By assumption $\prb{\|\bm{e}(0{:}N)\|\leq \frac{q}{C}}\geq 1-\theta$. Also, by the Lipschitz condition, it holds that  
    \begin{align*}
    &\mathrm{Pr} \left\{ \left| \rho^\phi(\bm{z}(0{:}N) + \bm{e}(0{:}N)) - \rho^\phi(\bm{z}(0{:}N)) \right| \right. \\
    &\quad \left. \leq L_\phi \left\| \bm{e}(0:N) \right\| \leq L_\phi   \frac{q}{C} \right\} \geq 1-\theta.
\end{align*}
    Focusing on the case where $\rho^\phi(\bm{z}(0{:}N)+\bm{e}(0{:}N))\leq\rho^\phi(\bm{z}(0{:}N))$, the previous Lipschitz inequality yields $$\prb{\rho^\phi(\bm{x}(0{:}N))\geq \rho^\phi(\bm{z}(0{:}N))-L_\phi   \frac{q}{C}\geq 0}\geq 1-\theta,$$ since $\rho^\phi(\bm{z}(0{:}N))\geq L_\phi   \frac{q}{C}$, completing the proof.
\end{proof}

\subsection{Relaxation of the chance-constrained synthesis problem}

% Next, we denote Lipschitz constants of the robustness functions of formulas $\phi_\nu$, $\nu\in\cK_\phi$, as $L_{\phi_\nu}$, $\nu\in\cK_\phi$. 
We now present a data-driven relaxation of \eqref{eq:multi_agent_problem}.

\begin{theorem}\label{thm:relaxed_prob}
    Let $L_{\phi_\nu}$, $\nu\in\cK_\phi$, be the Lipschitz constants of the robustness functions of formulas $\phi_\nu$, $\nu{\in}\cK_\phi$. Let also test aggregate error trajectories $\bm{e}^{(0)}_\nu(0{:}N)$, $\nu{\in}\cK_\phi$, with $e_\nu(0)=0$, satisfy \eqref{eq:chance_constrained_error_bounds}, and aggregate deterministic trajectories $\bm{z}_\nu(0{:}N)$, $\nu{\in}\cK_\phi$, with $z_\nu(0)=x_\nu(0)$, driven by feedforward controllers $\bm{v}_\nu(0{:}N{-}1)$, $\nu{\in}\cK_\phi$, be solution to:
    % Consider the optimization in \eqref{eq:multi_agent_problem}, the dynamics in \eqref{eq:clique_decomposed_dynamics}, and the PRs in \eqref{eq:chance_constrained_error_bounds}. 
    % Consider the MAS with dynamics in \eqref{eq:MAS} subject to the STL formula in \eqref{eq:global_phi}. Let disturbance feedback gains in $\bm{\Gamma}^\ast=\{\bm{\Gamma}_1^\ast,\ldots,\bm{\Gamma}_M^\ast\}$, weights in $C^\ast=\left\{\{C_{\nu}^\ast\}_{\nu\in\cK_\phi}\right\}$ and the quantile parameter $q$ be obtained as in  
    % Proposition \ref{prop:prediction_regions}. 
    % Let $L_{\phi_\nu}$, $\nu\in\cK_\phi$, be the Lipschitz constants of the robustness functions of formulas $\phi_\nu$, $\nu\in\cK_\phi$, and % Let $\bm{z}_\nu(0:N)$ (and $\bm{e}_\nu(0:N)$), $\nu\in \cK_\phi$, be nominal (and error) trajectories corresponding to formulas $\phi_\nu$, $\nu\in\cK_\phi$. 
    % assume 
    \begin{align}
    \operatorname*{Minimize}_{\substack{\bm{v}(0),\; \bm{z}(0)}} \;&\sum_{i=1}^M\left(\sum_{t=0}^{N-1}(\ell_i(z_i(t), v_i(t))) + V_{f,i}(z_i(N))\right)
         \nonumber \\
     \mathrm{subject\;to~} &  z(t+1) = Az(t)+Bv(t),\; t\in\N_{[0,N)}, \nonumber \\
         %&\rho^{\phi_i}(\bm{z}_i(0:N))\geq L_{\phi_i}   \frac{q}{C_i}, \; i\in \cV, \label{eq:ith_STL_constraint_tight}\\
         &  \rho^{\phi_\nu}(\bm{z}_\nu(0{:}N))\geq L_{\phi_\nu}  \frac{q}{C_\nu^\ast}, \; \nu\in \cK_\phi. \label{eq:multi_agent_problem_tight}
    \end{align}   
    Then, $\prb{\bm{x}^{(0)}_\nu(0{:}N)\models \phi_\nu,\; \forall \nu\in\cK_\phi}\geq 1-\theta$, with $\bm{x}^{(0)}_\nu(0{:}N)=\bm{z}_\nu(0{:}N)+\bm{e}_\nu^{(0)}(0{:}N)$, $\nu\in\cK_\phi$.
    % has a feasible solution $\bm{v}(0{:}N{-}1){=}(v(0),\ldots,v(N{-}1))$. Then, $\bm{u}(0{:}N{-}1)$, where $u(t){=}(u_1(t),\ldots,u_M(t))$, with $u_i(t){=}\sum_{k=0}^{t-1}\Gamma_i^{\ast,t,k}w_i(k){+}v_i(t)$, $i\in\cV$, $t\in\N_{[0,N)}$ together with the corresponding $\bm{x}(0{:}N)$ is a feasible solution for \eqref{eq:multi_agent_problem}. %, that is
    % \begin{equation}
    %     \prb{\bm{x}(0:N)\models \phi} \geq 1-\theta.
    % \end{equation}
\end{theorem}
\begin{proof}
    The result follows from the condition in~\eqref{eq:chance_constrained_error_bounds}, 
which can be equivalently written as 
\(\Pr \big\{ \|\bm{e}_\nu(0{:}N)\| \leq \sfrac{q}{C_\nu^\ast}, \; \nu \in \mathcal{K}_\phi \big\} \geq 1{-}\theta\), 
which, by Lemma~\ref{lemma:tightening_single_formula} and the feasibility of the problem in~\eqref{eq:multi_agent_problem_tight}, 
implies \(\Pr\{ \rho^{\phi_\nu}(\bm{x}^{(0)}_\nu(0{:}N)) \geq 0, \; \nu \in \mathcal{K}_\phi \} \geq 1{-}\theta\), or, equivalently, \(\Pr(\bm{x}^{(0)}_\nu(0{:}N) \models \phi_\nu, \; \nu \in \mathcal{K}_\phi ) \geq 1{-}\theta\), which completes the proof.
\end{proof}

% We remark that to handle STL constraints as in \eqref{eq:multi_agent_problem_tight}, one can leverage existing STL methods and toolboxes that utilize either integer programming using binary variables or nonlinear solvers using log-sum-exp underapproximations of the STL robustness function \cite{Donze2015}. Due to space limitations, we defer a detailed discussion of these technical aspects to an extended version of this work. Next, we decompose \eqref{eq:multi_agent_problem_tight} into individual agent-level problems to improve tractability.

\subsection{Distributed control synthesis}

To handle STL constraints in \eqref{eq:multi_agent_problem_tight}, one may use existing methods based on integer programming or log-sum-exp approximations \cite{Donze2015}. We defer details to an extended version. Next, we propose an agent-level decomposition of \eqref{eq:multi_agent_problem_tight} and an iterative procedure to manage multi-agent complexity.
  
%, extending the practicality of the multi-agent framework at hand.  
% First, we assume the following.
% \begin{assumption}\label{ass:global_problem_tight}
% % Let $t$-PRSs, $E(t)$, $t\in\N_{[0,N]}$, for \eqref{eq:clique_error_dyn}, and the gain, $K$, for the pair $(A,B)$, be as constructed in Thm. \ref{thm:deterministic_problem}. 
% The optimization \eqref{eq:multi_agent_problem_tight} has a feasible solution $\bm{v}(0)=(v(0),\ldots,v(N-1))$, $\bm{z}(0)=(z(0),\ldots,z(N))$.
% \end{assumption}

\subsubsection{Agent-level decomposition} For agent $i$ participating in at least one clique, i.e., $i\in\nu$, with $\nu\in\cK_\phi$, let $\cT_i = \{\nu \in \cK_\phi \mid i \in \nu \}$ be the set of cliques containing $i$. % by
% \begin{equation}\label{eq:Ti}
%     \cT_i = \{\nu \in \cK_\phi \mid \nu \ni i\}.
% \end{equation}
% where $\operatorname{cl}(i) = \{\nu \in \cK_\phi,\; \nu \ni i\}$ is the set of cliques that contain $i$. Let $j\in \cT_i$, with $j=(i_1,\ldots,i_{|j|})$. Let a trajectory $\bm{z}_{ij}(0)=(z_{ij}(0),\ldots,z_{ij}(N))$, where $z_{ij}(t)=(z_{i_1}(t),\ldots,z_i(t),\ldots,z_{i_{|j|}}(t))$, with $t\in\N_{[0,N]}$, and the order $i_1<\ldots<i<\ldots<i_{|j|}$ being specified by the lexicographic ordering of the node set $\cV=\N_{[1,M]}$. 
Then, a multi-agent STL formula equivalent to the original one in \eqref{eq:global_phi} is defined as $\hat{\phi}=\bigwedge_{i\in \cV}\hat{\phi}_i$, where $\hat{\phi}_i=\bigwedge_{\nu_i\in\cT_i}\phi_{\nu_i}$. In the following, we denote $\varrho^{\phi_{\nu_i}}(\bm{z}_{\nu_i}(0{:}N)){=}\rho^{\phi_{\nu_i}}(\bm{z}_{\nu_i}(0{:}N)){-}L_{\phi_{\nu_i}} \frac{q}{C_{\nu_i}}$, $\nu_i\in\cT_i$, and introduce the agent-level problems $P_i^0$ and $P_i^t$, $t\in\N_{[1,N]}$, for the $i$th agent, where
\begin{subequations}\label{eq:initial_local_problem}
\begin{align}
    P_i^0&:=\operatorname*{Minimize}_{\substack{\bm{v}_i^0, \bm{z}_i^0}}\;\cL_i(\bm{z}_i^0,\bm{v}_i^0)\; \mathrm{subject\;to} \label{eq:cost_intro} \\
     & \;\;\;\;\;\; z_i^0(k+1) {=} A_iz_i^0(k){+}B_iv_i^0(k), \; k\in\N_{[0,N)}, 
     \label{eq:local_STL_dynamics}  
     \\ 
     &\;\;\;\;\;\; \varrho^{\phi_i}(\bm{z}_i^0)\geq 0, \; \mathrm{with}\; z_i^0(0)=x_i(0), \label{eq:local_STL_initial}  
     % &z_i(0)=x_{0,i}\label{eq:initial_agent_i_dynamics}
\end{align}    
\end{subequations}
and
\begin{subequations}\label{eq:kth_local_problem}
\begin{align}
    P_i^t&:=\operatorname*{Minimize}_{\substack{\bm{v}_i^t, \bm{z}_i^t}}\; \cL_i(\bm{z}_i^t,\bm{v}_i^t)-\Omega_i\mu_{\nu_t}^t \; \mathrm{subject\;to~} \label{eq:local_problem_kth_iter} \\
     & \;\;\;\;\;\; z_i^t(k+1) = A_iz_i^t(k)+B_iv_i^t(k), \; k\in\N_{[t,N)},\\ 
     & \;\;\;\;\;\; \varrho^{\phi_i}(\bm{z}_i^t)\geq 0, \; \mathrm{with}\; z_i^t(t)=x_i(t), \label{eq:kth_agent_i_local_task}\\  
     &\;\;\;\;\;\; \varrho^{\phi_{\nu_t}}(\bm{z}_{\nu_t}^t) {\geq} \mu_{\nu_t}^t, \; \nu_t {=} \argmin_{\nu\in \cT_i, \; |\nu|>1} \{\varrho^{\phi_{\nu}}(\bm{z}_{\nu}^{t-1})\}, \label{eq:least_violating_joint_task_a}\\
     &\;\;\;\;\;\; \mu_{\nu_t}^t\geq \min\left(0,\varrho^{\phi_{\nu_t}}(\bm{z}_{\nu_t}^{t-1})\right), \label{eq:least_violating_joint_task_b}\\
     &\;\;\;\;\;\; \varrho^{\phi_{\nu}}(\bm{z}_{\nu}^t){\geq} \min\left(0,\varrho^{\phi_{\nu}}(\bm{z}_{\nu}^{t-1})\right), \forall \nu {\in} \cT_i{\setminus} \{\nu_t,i\} \label{eq:robust_joint_tasks}
\end{align}    
\end{subequations}
where $\Omega_i\gg 0$, $z_i^t(\tau)$ denotes the prediction of $x_i(\tau)$ carried out at time $t$, $\varrho^{\phi_{\nu}}(\bm{z}_{\nu}^t)$ is the tightened robustness function of the formula $\phi_{\nu}$, $\nu\in\cT_i$, evaluated over the trajectory $\bm{z}_{\nu}^t$, and
\begin{align}
    \bm{z}_\nu^t&{=}(x_\nu(0),...,x_\nu(t{-}1),z_\nu^t(t),...,z_\nu^t(N)),\\
    \bm{v}_i^t&{=}(v_i(0),...,v_i(t{-}1),v_i^t(t),...,v_i^t(N{-}1)),\\
    \cL_i(\bm{z}_i^t,\bm{v}_i^t)&{=}\sum_{k=0}^{N-1}\ell_i(z_i^t(k),v_i^t(k)){+}V_{f,i}(z_i^t(N)).
\end{align}

$P_i^0$ and $P_i^t$ decompose the centralized problem in \eqref{eq:multi_agent_problem_tight}: $P_i^0$ is solved by all agents at $t=0$, while for $t \geq 1$, $P_i^t$ is solved by a subset of agents. Details are in \cite{VlahakisCDC24}, where these subsets are scheduled offline in a centralized way, which we now extend to distributed coordination.

\subsubsection{Distributed MPC-like implementation} Let 
\begin{align}
    \bm{z}_i^t(x_i(t),\bm{v}_i^t){=}(x_i(0),...,x_i(t),z_i^t(t+1),...,z_i^t(N)),
\end{align}
denote a trajectory where the last $N{-}t$ states are generated by the last $N{-}t$ inputs of $\bm{v}_i^t$ starting from $x_i(t)$. For $t > 1$, the $i$th agent computes $\rho^{\phi_{\nu_i}}(\bm{z}_{\nu_i}^{t-1})$, $\nu_i \in \cT_i$, and estimates the robustness function of $\hat{\phi}_i$ as $r_i^t=\min_{\nu_i \in \cT_i}\left( \varrho^{\phi_{\nu_i}}(\bm{z}_{\nu_i}^{t-1}) \right)$. Agent-$i$ communicates $r_i^t$ and $\bm{z}_i^t(x_i(t),\bm{v}_i^{t-1})$ to all agents in $\nu_i$, $\nu_i \in \cT_i$, and receives their corresponding information. Agent-$i$ either solves $P_i^t$ if $r_i^t =\min_{j\in\nu_i,\;\nu_i\in\cT_i}\left( r_j^t\right)$ or retains its input sequence from $t{-}1$, that is $\bm{v}_i^t{=}\bm{v}_i^{t-1}$. Alg.~\ref{alg:control_syn} summarizes this distributed strategy, and Prop.~\ref{prop:mpc} highlights its benefits. 

\begin{proposition}\label{prop:mpc}
    Suppose each agent-$i$ solves $P_i^0$ at $t{=} 0$ and executes Alg. \ref{alg:control_syn} from $t{\geq}1$. Let $\cO_t{\subset} \cV$ collect indices of agents solving $P_i^t$ at $t {\geq} 1$, and assume $P_i^0$, $i\in\cV$, $P_i^t$, with $i{\in}\cO_t$ for $t{\geq} 1$, are feasible. It holds: 1) if $i,j{\in}\cO_t$ for some $1 {\leq} t {\leq} N$, then $\cT_i{\cap} \cT_j{=}\emptyset$, %2) $\prb{\rho^{\phi_i}(\bm{x}_i(0{:}N)){\geq} 0}{\geq} 1{-}\theta$,  
    and 2) collaborative tasks $\phi_\nu$, $\nu\in\cK_\phi$, are minimally violated or satisfied with probability at least $1{-}\theta$.
\end{proposition}
\begin{proof}
    1) By assumption, at \( t \geq 1 \), \( r_i^t < r_l^t \) for all \( l \in \nu_i \), where \( \nu_i \in \mathcal{T}_i \), and \( r_j^t < r_s^t \) for all \( s \in \nu_j \), where \( \nu_j \in \mathcal{T}_j \). Assuming without loss of generality that \( r_i < r_j \), if \( \exists \nu \in \mathcal{T}_i \cap \mathcal{T}_j \), then \( j \notin \mathcal{O}_t \), which contradicts the assumption.
    % 2) This follows from the feasibility of $P_i^0$, the recursive feasibility assumption and the tighter robustness function in \eqref{eq:kth_agent_i_local_task}, and by Prop. \ref{prop:prediction_regions}.
    2) This follows from the feasibility of the constraints in \eqref{eq:least_violating_joint_task_a}--\eqref{eq:least_violating_joint_task_b}, which ensures the improvement of the most violated (least robust) joint task \( \phi_{\nu_t} \), and from the feasibility of the constraint in \eqref{eq:robust_joint_tasks}, which ensures non-violation or improvement of other joint tasks. %The \(\min\) operator in \eqref{eq:least_violating_joint_task_b}--\eqref{eq:robust_joint_tasks} relaxes constraints for already satisfied tasks.
    The probability guarantee follows from the tighter robustness functions and Prop. \ref{prop:prediction_regions}.
\end{proof}

\begin{algorithm}
\caption{Distributed STL control of agent-$i$}
\label{alg:control_syn}
\begin{algorithmic}[1]
\For{$t$ in $1:N$}
    \State \textbf{Compute} $r_i^t=\min_{\nu_i \in \cT_i}\left( \varrho^{\phi_{\nu_i}}(\bm{z}_{\nu_i}^{t-1}) \right)$ 
    % \State \textbf{Measure} $x_i(t)$ and $w_i(t{-}1)$
    % \State \textbf{Construct} $\bm{z}_i^t(x_i(t),\bm{v}_i^{t-1})$
    \State \textbf{Communicate} $r_i^t$, $\bm{z}_i^t(x_i(t),\bm{v}_i^{t-1})$ to $j\in\nu_i$, $\nu_i\in\cT_i$
    \State \textbf{Receive} $r_j^t$, $\bm{z}_j^t(x_i(t),\bm{v}_i^{t-1})$ from  $j\in\nu_i$, $\nu_i\in\cT_i$
    
    \If{$r_i^t<r_j^t$ for all $j\in\nu_i$, $\nu_i\in\cT_i$}
        \State \textbf{Solve} $P_i^t$ and store $(\bm{v}_i^t,\bm{z}_i^t)$  
    \Else
        \State \textbf{Update} $\bm{v}_i^t\leftarrow\bm{v}_i^{t-1}$ and $\bm{z}_i^t\leftarrow \bm{z}_i^t(x_i(t),\bm{v}_i^{t-1})$
    \EndIf

    \State \textbf{Apply} $u_i(t)=\sum_{k=0}^{t-1}\Gamma_i^{t,k}w_i(k)+v_i^t(t)$
\EndFor
%\State \textbf{return} $(\bm{v}(0),\bm{z}(0))$ \label{algo:return}
\end{algorithmic}
\end{algorithm}

\section{Example}\label{sec:example}

We revisit the 10-agent setup from \cite[Sec. IV]{VlahakisCDC24}, where agents with single-integrator dynamics are assigned individual and collaborative STL tasks under Gaussian disturbances $w_i(t) \sim \cN(0, 0.05I_2)$, $i \in \N_{[1,10]}$. We refer to \cite[Sec. IV]{VlahakisCDC24} for detailed description. Each agent must visit regions $T_i$ and $G_i$ within $\N_{[10,50]}$ and $\N_{[70,100]}$, respectively, while avoiding obstacles in $\cX$. Collaborative tasks require clique members to rendezvous within $\N_{[0,100]}$. We recall the clique set \(
\cK_\phi = \{(1),\ldots,(10), (1,2,3), (1,5), (3,4), (4,5), (5,6), (4,7), (6,\allowbreak8), (6,9), (7,8), (8,10), (9,10)\}.
\) The STL specification spans 100 steps and must be satisfied with probability $95\%$. 

\textbf{PR-scaling and feedback design:} We generate 100 training and 100 calibration disturbance sequences (length 100) and run Alg.~\ref{alg:CGamma_train} for $\tau_{\max} = 4$ to compute $(C^\ast, \bm{\Gamma}^\ast)$ using GUROBI. Average solve times for $P(C)$ and $P(\bm{\Gamma})$ are under 1 sec and 20 min, respectively (Intel i7, 32 GB RAM). Using the calibrated disturbance data and Prop.~\ref{prop:prediction_regions}, we compute PRs $\B_\infty(q/C^\ast_\nu)$ with $0.7398 \leq q/C^\ast_\nu \leq 0.8985$, $\nu \in \cK_\phi$.

\textbf{Distributed STL synthesis:} To meet the $95\%$ satisfaction target, we tighten the underlying robustness functions using Thm.~\ref{thm:relaxed_prob} and Prop.~\ref{prop:prediction_regions}, with Lipschitz constants ranging from $0.046$ to $1$ across cliques (Lemma~\ref{lemma:lipschitz}). We solve $P_i^0$ for all $i$ and apply Alg.~\ref{alg:control_syn}, where agents iteratively solve $P_i^t$, $t \in \N_{[1,100]}$, exchanging estimates of robustness functions within cliques. Fig.~\ref{fig:trajectories} shows results from one experiment. Out of 100 runs, 97 satisfied the STL formula. Fig.~\ref{fig:bounds_hist} (right) shows agent selection frequency during coordination. Our code is available at \cite{vlahakis2025cdc_github}.

\textbf{Comparison with \cite{VlahakisCDC24}:} The method in \cite{VlahakisCDC24} relies on a union-bound-based uncertainty quantification providing conservative PRs for the error trajectories, which limits the specified satisfaction probability to $70\%$. In contrast, the proposed data-driven approach yields considerably tighter PRs with high confidence, depending on the number of available samples. Fig.~\ref{fig:bounds_hist} (left) compares the $70\%$ $l_2$ PR (red circle) from \cite{VlahakisCDC24} with the data-driven $95\%$ $l_\infty$-based PR obtained using $10^4$ new calibration samples and a tightened $97\%$ quantile, guaranteeing over $99.9999\%$ confidence $1-\beta$ (see Rem.~\ref{rem:CP}). Unlike the offline design in \cite{VlahakisCDC24}, the new method enables fully distributed coordination via Alg.~\ref{alg:control_syn}.
 %Future work will address execution of \( P_i^t \) with input constraints. 

\begin{figure}
    \centering
    \input{figures/10agents_traj_3April}
    \caption{Simulation under Alg.~\ref{alg:control_syn} with $x_i(0)$ as crosses, $T_i$ areas as diamonds, $G_i$ areas as boxes, and 3 brown obstacles.}
    \label{fig:trajectories}
\end{figure}

\begin{figure}[ht]
    \centering
    \begin{minipage}[b]{0.53\linewidth}
        \centering
        \input{figures/error_bounds_CDC24_25_compare}
        % \caption{Ten noisy multi-agent instances resulting from Alg.~\ref{alg:control_syn}.}
        % \label{fig:figure1}
    \end{minipage}
    \hfill
    \begin{minipage}[b]{0.43\linewidth}
        \centering
\definecolor{mycolor1}{rgb}{0.20000,0.60000,0.80000}%

\begin{tikzpicture}
[font=\tiny]
\begin{axis}[%
width=0.81in,
height=0.8in,
at={(0,0)},
scale only axis,
bar shift auto,
xmin=-0.2,
xmax=11.2,
xtick={ 1,  2,  3,  4,  5,  6,  7,  8,  9, 10},
xlabel style={font=\small, color=white!15!black, yshift=0.2cm},
xlabel={Agents' indices},
ymin=0,
ymax=18,
ylabel style={font=\small, color=white!15!black, yshift=-0.25cm},
ylabel={Counts},
axis background/.style={fill=white},
title style={font=\bfseries},
xmajorgrids,
ymajorgrids
]
\addplot[ybar, bar width=0.2, fill=red, draw=red, area legend] table[row sep=crcr] {%
% 1	8.92857142857143\\
% 2	8.63095238095238\\
% 3	7.44047619047619\\
% 4	10.1190476190476\\
% 5	9.52380952380952\\
% 6	10.4166666666667\\
% 7	10.7142857142857\\
% 8	10.1190476190476\\
% 9	12.5\\
% 10	11.6071428571429\\
1 9\\
2 13\\
3 11\\
4 16\\
5 15\\
6 16\\
7 17\\
8 16\\
9 10\\
10 13\\
};
\addplot[forget plot, color=white!15!black] table[row sep=crcr] {%
-0.2	0\\
11.2	0\\
};
\end{axis}
\end{tikzpicture}%
        % \caption{Histogram of agent selections according to Alg. \ref{alg:control_syn}.}
        % \label{fig:bounds_hist}
    \end{minipage}
    \caption{Left: Red circle shows the PR for $\bm{e}_i(t)$ from \cite{VlahakisCDC24}, computed via a union bound with $70\%$ probability; blue box denotes the proposed $95\%$ $\ell_\infty$-based PR with confidence $1{-}\beta{\gg} 99.99\%$ (see Prop.~\ref{prop:prediction_regions} and Rem.~\ref{rem:CP}). 
Right: Histogram of agents during the experiment of Fig.~\ref{fig:trajectories} under Alg.~\ref{alg:control_syn}.
}
    \label{fig:bounds_hist}
\end{figure}

%
% \begin{figure}[htbp]
%     \centering
%     \includegraphics[width=\columnwidth, height=4cm]{figures/10agents_trajectories}
%     \caption{Ten noisy multi-agent instances resulting from Alg.~\ref{alg:control_syn}.}
%     \label{fig:trajectories}
% \end{figure}
% \begin{figure}[htbp]
%     \centering
%     \input{figures/histogram_agents}
%     \caption{Histogram of agent selections according to Alg. \ref{alg:control_syn}.}
%     \label{fig:histogram_agents}
% \end{figure}

\section{Conclusion}\label{sec:concl}

We propose distribution-free control synthesis for stochastic linear MAS under chance-constrained collaborative STL specifications. Exploiting linearity, we decompose each agent’s dynamics into deterministic and error components, following a disturbance-feedback (DF) and feedforward control policy. We design DF via CVaR-based optimization on training samples and use conformal prediction on calibration data to quantify prediction regions for error trajectories, enabling us to address chance constraints through Lipschitz tightening. This relaxation yields a centralized deterministic problem, whose solution provides the feedforward inputs. To improve scalability, we decompose it into agent-level MPC subproblems, resulting in a distributed control architecture. Compared with \cite{VlahakisCDC24}, the approach provides tighter uncertainty bounds and enhances scalability via distributed coordination.

\balance

\bibliographystyle{IEEEtran}

\bibliography{biblio}

\end{document}